\documentclass[submission,copyright,creativecommons]{eptcs}
\providecommand{\event}{EXPRESS 2019} 
\usepackage{breakurl}             
\usepackage{underscore}           

\usepackage{microtype}
\usepackage{amsmath,latexsym,amssymb,amsthm}
\usepackage[dvipsnames]{xcolor}

\newif\iffull\fullfalse

\newtheorem{example}{Example}

\newtheorem{definition}{Definition}
\newtheorem{lemma}{Lemma}
\newtheorem{proposition}{Proposition}

\usepackage[all]{xy}

\iffull

\else

\fi

\newcommand{\define}{\triangleq}

 \newcounter{ncomm}%

\newcommand{\CS}{CS}

\newcommand{\A}{{\cal A}}
\newcommand{\E}{{\cal E}}
\newcommand{\F}{{\cal F}}

\newcommand{\Conf}[1]{{\it Conf}(#1)}
\newcommand{\STr}[1]{{\it SeqTr}(#1)}
\newcommand{\StTr}[1]{{\it StepTr}(#1)}
\newcommand{\POM}[1]{{\it Pom}(#1)}
\newcommand{\pom}[1]{{\it pomset}(#1)}
\newcommand{\poset}[1]{{\it poset}(#1)}
\newcommand{\red}[1]{\stackrel {#1} \longrightarrow }

\newcommand{\iteq}{\approx_{\rm it}}
\newcommand{\steq}{\approx_{\rm st}}
\newcommand{\pteq}{\approx_{\rm pt}}
\newcommand{\ibeq}{\approx_{\rm ib}}
\newcommand{\sbeq}{\approx_{\rm sb}}
\newcommand{\pbeq}{\approx_{\rm pb}}
\newcommand{\hbeq}{\approx_{\rm hb}}
\newcommand{\whbeq}{\approx_{\rm whb}}
\newcommand{\hhbeq}{\approx_{\rm hhb}}

\title{Conflict vs Causality in Event Structures}
\author{Daniele Gorla, Ivano Salvo, Adolfo Piperno
\institute{Sapienza University of Rome, Dpt. of Computer Science}
\email{\{gorla,salvo,piperno\}@di.uniroma1.it}
}
\def\titlerunning{Conflict vs Causality in Event Structures}
\def\authorrunning{D. Gorla, I. Salvo and A. Piperno}

\begin{document}

\sloppy

\maketitle

\begin{abstract}
Event structures are one of the best known models for concurrency. 
Many variants of the basic model and many possible notions of equivalence for them
have been devised in the literature. 
In this paper, we study how the spectrum of equivalences for Labelled Prime Event Structures
built by Van Glabbeek and Goltz changes if we consider two simplified notions of event structures: 
the first is obtained by removing the causality relation (Coherence Spaces) and 
the second by removing the conflict relation (Elementary Event Structures).
As expected, in both cases the spectrum turns out to be simplified, since some notions of equivalence 
coincide in the simplified settings; actually, we prove that removing
causality simplifies the spectrum considerably more than removing conflict.
Furthermore, while the labeling of events and their cardinality play no role when removing
causality, both the labeling function and the cardinality of the event set dramatically influence the spectrum of
equivalences in the conflict-free setting. 
\end{abstract}

\section{Introduction}

Event structures \cite{NielsenPW81,WN95} are one of the best known models for concurrency.
Basically, they are collections of possible events, some of which are conflicting
(i.e., the execution of an event forbids the execution of other events), while 
others are causally dependent (i.e., an event cannot be executed if
it has not been preceded by other ones). 
Prime Event Structures (written PESs) are the earliest and simplest form of event structure, where causality is a partial order and conflict between events is inherited by their causal successors.
Events are often labelled with actions, to represent different occurrences of the same action. In this paper, we shall focus on labelled PESs, referring to them simply as PESs, for the sake of simplicity.

Conflict and causality are fundamental concepts 
for concurrency; indeed, they can
also be found in other well-established models for concurrent computation, like Petri nets 
\cite{Pet77,Pet81,Petri}
and process algebras
\cite{BK84,CSP78,Mil89}
(where they are called choice and sequential composition, respectively). Not incidentally,
both conflict and causality influence the evolution of an event structure, whose semantics is given by means of 
{\em configurations}: these are finite conflict-free subsets of events that are closed by causal predecessors. 
Configurations take note of the events occurred so far during a computation.
Indeed, starting from the empty configuration, the evolution of an event structure is obtained by 
selecting one or more events that are
causally enabled by the events executed so far, and non-conflicting with any of them.
However, not all sets of events can
be simultaneously executed: this yields the derived notion of {\em concurrent} events, that are those
that are neither in conflict nor causally dependent from one another.

A fruitful research line is the study of different possible notions of equivalence
for event structures, inspired by the richness of equivalences for process algebras \cite{G90,G93}.
Indeed, apart from the classical distinction between trace and bisimulation-based equivalences,
in the framework of PESs many features can be observed to distinguish two event
structures. In this paper, we follow \cite{GG01} and consider the following equivalences:
\begin{enumerate}
\item {\em interleaving trace and bisimulation equivalences} (written $\iteq$ and $\ibeq$):
these are the direct counterparts of trace and bisimulation equivalence for process algebras  \cite{CSP78,Mil89}; 
in the framework of PESs, only (the label of) one single event at a time is observed, either in a sequence forming a trace
or in the bisimulation game based on coinduction.

\item {\em step trace and bisimulation equivalences} (written $\steq$ and $\sbeq$) \cite{Pomello85},
where the units of observation are sets of concurrent (and causally enabled) events
To be more precise, we do not observe sets of events but the multisets of the labels associated to
the selected events (recall that the same label can be given to different events).

\item {\em pomset trace and bisimulation equivalences} (written $\pteq$ and $\pbeq$) \cite{BoudolC87},
where the units of observation are \emph{sets of events} together with their causality and concurrency relations;
again, since different events can have the same label, a set of events generates a partially
ordered multiset (hence, the name {\em pomset}), based on the causality relation.

\item different variants of {\em history preserving bisimulation}, where the configurations
of the two PESs related by a bisimulation must have the same causal dependencies.
According to how this requirement is formalized, we have:
\begin{enumerate}
\item {\em weak history preserving bisimulation} (written $\whbeq$) \cite{DDM87}, where
every pair of configurations is formed by isomorphic (w.r.t. their causal dependencies) pomsets;
\item {\em history preserving bisimulation} (written $\hbeq$) \cite{DDM88,RT88}, where
every pair of configurations is formed by isomorphic (w.r.t. their causal dependencies) pomsets
and the isomorphism grows during the computation (whereas, for $\whbeq$ two consecutive
pairs of configurations could be related by totally different isomorphisms);
\item {\em hereditary history preserving bisimulation} (written $\hhbeq$) \cite{bedna01}, which is $\hbeq$ 
with the additional requirement that the isomorphism is maintained also when going
back in the computation.
\end{enumerate}
\end{enumerate}
These 9 equivalences, together with PES isomorphism $\cong$, form a well known spectrum \cite{Fecher04,GG01} that we depict in Figure \ref{eq-pes}
(where the term {\em autoconcurrency} means existence of a configuration
containing two different concurrent events with the same label).

Orthogonally, since their birth, many variants of the basic framework have appeared in the literature.
The basic model has been both extended with more sophisticated features and simplified by
removing features. Richer notions of event structures include, among the others,
flow event structures \cite{BC88}, stable/non-stable event structures \cite{Winskel86}
and configuration structures \cite{GP95}.
By contrast, simplified models are obtained either by removing the causality relation,
yielding {\em coherence spaces} \cite{Girard87} (written \CS s in this paper), 
or by removing the conflict relation, yielding
{\em elementary event structures} \cite{NielsenPW81} (written EESs). 
Both these models have interesting applications
in the literature: the former one is used for giving the semantics 
of linear logic \cite{Girard87} and typed lambda-calculus
\cite{BE91,BE91b}; the latter one is a common variant of PESs
(\cite{NielsenPW81,NT02,GG01}, just to cite a few).

\begin{figure}[t]
\hspace*{-.8cm}
	\begin{tabular}{cc}
	    \begin{tabular}{c}
		\begin{minipage}{0.46\textwidth}
		\center\ 
                    {\hfill\framebox{\qquad\quad
                    $
                    \xymatrix@R=20pt{
                    & \iteq
                    \\
                    \ar[ur]
                    \ibeq
                    &
                    \ar[u]
                    \steq
                    \\
                    \ar[u] \ar[ur]
                    \sbeq 
                    &
                    \pteq \ar[u]
                    \\
                    \ar[u] \ar[ur]
                    \pbeq
                    &
                    \ar[ul]\ar[u]
                    \whbeq
                    \\
                    \ar@{-->}[ur]
                    \ar[u]
                    \hbeq
                    \\
                    \ar[u]
                    \hhbeq
                    \\
                    \ar[u]
                    \cong
                    }
                    $\qquad\qquad
                    }\hfill}
                    \vspace*{-.2cm}
                    \caption{The spectrum of equivalences for PESs
                    {\footnotesize (`$\rightarrow$' means `$\subset$'; 
                    `$\dashrightarrow$' means `$=$', if no autoconcurrency is present,
                    and means `$\subset $', otherwise)}}
                    \label{eq-pes}
		\end{minipage}
                    \vspace*{.5cm}
	    \\
		\begin{minipage}{0.46\textwidth}
		\center
                {\hfill\framebox{
                $
                \xymatrix@R=20pt{
                \iteq\ =\ \steq\ =\ \pteq
                \\
                \ibeq\ =\ \sbeq\ =\ \pbeq\ =\ \whbeq\ =\ \hbeq
                \ar[u]
                \\
                \hhbeq
                \ar[u]
                \\
                \cong
                \ar[u]
                }
                $
                }\hfill}
                \vspace*{-.2cm}
                \caption{The spectrum for \CS s}
                \label{eq-sharp}
		\end{minipage}
  	    \end{tabular}
&
	       \begin{tabular}{c}
		  \begin{minipage}{0.50\textwidth}
			\center
                {\hfill\framebox{
                $
                \xymatrix@R=3pt{
                & \iteq
                \\ 
                \\
                \\
                \ibeq \ar[uuur] & \steq \ar[uuu]
                \\ 
                \\
                \\
                \sbeq\ar[uuur]\ar[uuu]
                \\
                \\
                \\ 
                \pbeq\ =\ \pteq\ =\ \whbeq\ =
                \ar[uuu]
                \\ 
                \hbeq\ =\ \hhbeq\ =\ \cong
                }
                $
                }\hfill}
                \vspace*{-.3cm}
                \caption{The spectrum for finite EESs}
                \label{fig:finiteEESs}
		\end{minipage}
		\vspace*{.4cm}
		\\
		\begin{minipage}{0.50\textwidth}
			\center
                {\hfill\framebox{\qquad\quad
                $
                \xymatrix@R=18pt{
                & \iteq
                \\
                \ar[ur]
                \ibeq
                &
                \ar[u]
                \steq
                \\
                \ar[u] \ar[ur]
                \sbeq 
                &
                \pteq \ar[u]
                \\
                \ar[u] \ar[ur]
                \pbeq
                &
                \ar[ul]\ar[u]
                \whbeq
                \ar@{-->}@<-0ex>[l]^1
                \\
                \ar[u]
                \hbeq
                \ar@{-->}[ur]
                &
                \hspace*{-2.4cm}\,^2
                \\
                \ar@{-->}[u]^3
                \hhbeq
                \\
                \ar@{-->}[u]^4
                \cong
                }
                $\quad\qquad
                }\hfill}
                \vspace*{-.2cm}
                \caption{The spectrum for infinite EESs
                \footnotesize{(a numbered dashed arrow denotes an open question; for questions 2, 3 and 4, the arrow becomes solid if the question has a positive answer and becomes `$=$' otherwise; for question 1, the arrow disappears il the answer is positive and becomes solid otherwise).}}
                \label{infEESs}
		\end{minipage}
	    \end{tabular}
	\end{tabular}
\end{figure}

The aim of this paper is to investigate how the spectrum of Figure \ref{eq-pes}
changes when passing from PESs to \CS s and EESs. 
As expected, in both cases the spectrum turns out to be simplified, since some notions of equivalence 
coincide in the simplified settings. So, for every possible inclusion, we have to either
(1) prove that the inclusion becomes an equality, or (2) provide an example
in the simplified setting to distinguish the two equivalences (and confirm properness of the inclusion
also in the simplified setting).

The spectrum is radically
simplified in the framework of \CS s, as depicted in Figure \ref{eq-sharp}. As evident,
removing the causality relation reduces a complex lattice to a simple chain:
trace equivalences all coincide and represent the coarsest notion; 
they properly include bisimilarities (that all coincide, except for $\hhbeq$)
that in turn properly include the back-and-forth variant \cite{DMV90} of $\hbeq$.
Furthermore, the labeling function plays no role in such results; so, even the ``flattening'' labeling
(that associates the same action to every event) does not change the spectrum.

The situation is more articulated when conflict is removed, hence in the framework of
EESs. A posteriori, this is not surprising because a partial order (viz., the causality
relation) is a richer mathematical object than an irreflexive and symmetric relation
(viz., the conflict relation). What is really surprising is the fact that having finitely or 
infinitely many events makes a significant difference in terms of the distinguishing power
of  the studied equivalences; Figures \ref{fig:finiteEESs} and \ref{infEESs} give a visual
account of the difference. The first easy, but still interesting, result for finite EESs is
that $\pteq$, $\pbeq$, $\whbeq$, $\hbeq$, $\hhbeq$ and $\cong$ all coincide. 
This can be justified by observing that, being finite and without conflict,
the set of all the events of every such EES is a configuration of the EES itself; so, all notions
of equivalence that rely on some kind of pomset isomorphism collapse to EES isomorphism.
By contrast, for infinite EESs this does not hold anymore and some more inclusions
that were proper in Figure \ref{eq-pes} remain proper also in Figure \ref{infEESs}. 
%
Four questions remain open about strictness of some inclusions for infinite EESs.
However, even if the spectrum is not fully worked out, we have some
examples that let us claim that cardinality of the event set matters when only causality
is considered. By contrast, cardinality has no impact on the spectrum for \CS s.
Furthermore, we prove that restricting to ``flattening'' labeling functions makes $\iteq$ and $\ibeq$
collapse for EESs (again, in contrast with \CS s).

For all these reasons, our results seem to suggest that causality is a more foundational building block than conflict in event structures,
since it has a deeper impact on the discriminating power of equivalences for such models and
because it is more sensitive than conflict to issues like the cardinality of the set of events and their labeling.

The rest of the paper is organized as follows. In Section 2, we recall the basic definitions
and the spectrum for PESs, as reported in \cite{Fecher04}. Then, we move to consider
\CS s (Section 3) and EESs (Section 4); for the latter model, we also distinguish what happens
for finite (Section 4.1) and infinite structures (Section 4.2). Section 5 concludes the paper.

\section{Background: Prime Event Structures}
\label{model}

We start by summing up some well known notions from the theory of Event Structures~\cite{NielsenPW81}, by following the presentation in \cite{GG01}.

\begin{definition}[Prime Event Structures  \cite{NielsenPW81,WN95}]
A (labeled) {\em Prime Event Structure} (PES, for short) over an alphabet
$\A$ is a 4-tuple $\E = (E,\leq,\sharp,l)$ such that:
\begin{itemize}
\item $E$ is a set of {\em events};
\item $\leq\ \subseteq\, E\times E$ is the {\em causality} relation, i.e. a
partial order such that, for all $e \in E$, the set $\{e' : e' < e\}$ is finite;
\item $\sharp\ \subseteq\, E\times E$ is the {\em conflict} relation, i.e. an irreflexive
and symmetric relation such that, for all $e, e', e''\in E$, if $e < e'$ and $e \sharp e''$, then $e' \sharp e''$;
\item $l: E \rightarrow \A$ is the {\em labeling}\, function.
\end{itemize}
\end{definition}

Intuitively, $e' < e$ means that $e$ cannot happen before $e'$
(so, the execution of $e$ causally depends on the execution of $e'$),
whereas $e \sharp e'$ means that $e$ and $e'$ are mutually exclusive
(so, the execution of one prevents the execution of the other).
The condition $|\{e' : e' < e\}| < \infty$ ensures that every event can be
executed in a finite amount of time (i.e, after the execution of finitely many events).
Conflict inheritance (the condition in the third item of the previous definition)
is a sort of `sanity' condition, ensuring that every event inherits the conflicts of all its causal predecessors.
Finally, labels represent actions entailed by events, and so different events can have
the same label; this corresponds to the fact that the same action can occur different 
times during the execution of a system.

A derived notion is the {\em concurrency} relation, defined as follows:
$e\ co\ e'$ iff $(e,e') \not\in\ \leq\, \cup\, \geq\, \cup\ \sharp$.
When convenient, we shall write a PES by using the usual process algebra notation, where 
`$\parallel$' means `$co$', \linebreak `$;$' means `$<$'  and 
`$+$' means `$\sharp$'; moreover, we just write the labels, assuming
that the underlying events are all different.\footnote{
We remark that we shall use this syntax only when it comes handy to describe some particular 
PES in a succinct way;  in particular, in this paper we consider PESs as a per se semantic model, 
and not, e.g., as the interpretation domain for some process algebra. 
Furthermore, notice that PESs do not coincide with all the ESs that `$\parallel$', `$;$' and 
`$+$' can define: there are terms of this algebra that denote ESs that are not prime (e.g., $(a+b);c$) 
and there are PESs that are not definable using the given algebra 
(e.g., the event structure $\E$ in the proof of Prop. \ref{prop:st-i}) \cite{Gis84,Pratt85,Pra86}.
}

\begin{example}
\label{ex:PES-isom}
The expression $(a\parallel b) + (a;b)$
denotes the PES $\E=(E,\leq,\sharp,l)$ such that $E = \{e_1,e_2,e_3,e_4\}$, 
$e_i \leq e_i$ (for $i \in \{1,2,3,4\}$), $e_3 \leq e_4$,
$\sharp = \{(e_1,e_3),(e_2,e_3),(e_3,e_1),(e_3,e_2),$ $(e_1,e_4),(e_2,e_4),(e_4,e_1),(e_4,e_2)\}$,
$l(e_1) = l(e_3) = a$, and $l(e_2) = l(e_4) = b$.
\end{example}

To be precise,
$(a\parallel b) + (a;b)$ denotes the
$\cong$-class of the PES $\E$ given in Example \ref{ex:PES-isom}, 
where PES isomorphism is defined as follows.

\begin{definition}[PES isomorphism]
\label{def:PES-isom}
Let $\E = (E,\leq_E,\sharp_E,l_E)$ and $\F = (F,\leq_F,\sharp_F,l_F)$ be two PESs. We say that
$\E$ and $\F$ are {\em isomorphic}, and write $\E \cong \F$, if there exists a biiection $f: E \rightarrow F$
such that. for every $e,e' \in E$, it holds that :
\begin{itemize}
\item $e \leq_E e'$ if and only if $f(e) \leq_F f(e')$;
\item $e \sharp_E e'$ if and only if $f(e) \sharp_F f(e')$; and
\item $l_E(e) = l_F(f(e))$.
\end{itemize}
\end{definition}

Essentially, PES isomorphism only abstracts away from the set of events. So, for example, any $\F$
isomorphic to the PES $\E$ of Example \ref{ex:PES-isom} must be such that 
$F = \{e'_1,e'_2,e'_3,e'_4\}$ and $\leq/\sharp/l$ are defined as in Example \ref{ex:PES-isom},
but with $e'_i$ in place of $e_i$.

The semantics of a PES $\E$ is defined in terms of the possible states that the system modeled 
by the PES can pass through during its evolution, where such states are defined as follows.

\begin{definition}[Configurations]
A {\em configuration} of a PES $\E=(E,\leq,\sharp,l)$ is any $X \subseteq_{\rm fin} E$ 
such that 
\begin{itemize}
\item $e\ \sharp\ e'$, for every $e, e' \in X$; and
\item $\{e' : e' < e\} \subseteq X$, for every $e \in X$.
\end{itemize}
We denote with $\Conf\E$ the set of all configurations of $\E$.
\end{definition}

Configurations collect the events executed from the outset of the system; so, they
must be finite (they have to represent states reachable in a finite time), 
conflict-free (two conflicting events cannot be executed in the same system evolution)) 
and closed w.r.t. causal predecessors (an event can happen only if all its predecessors happened before).
For examples, the configurations of $\E$ from Example \ref{ex:PES-isom} are
$\emptyset, \{e_1\}, \{e_2\}, \{e_3\}, \{e_1, e_2\}, \{e_3,e_4\}$; notice that
$\{e_4\}$ is not a configuration because $e_4$ cannot stay in any configuration that 
misses its causal predecessor $e_3$,
and that $\{e_1,e_3\}$ is not a configuration because $e_1\,\sharp\,e_3$.

The way in which (the system modeled by) a PES evolves is usually given through 
some {\em labeled transition systems} (LTSs),
on top of which we can build different notions of equivalence between PESs. 
We now recall both the main transition relations and the main equivalences built on top of them.

The first transition relation between configurations states that
$X \red a X'$ whenever $X \subset X'$ and $X' \setminus X = \{e\}$, with $l(e) = a$; 
notation $X \red{}$ (resp., $X \not\!\!\red{}$) means that there exist $a$ and $X'$ (resp., no $a$ and $X'$)
such that $X \red a X'$. Coming back to Example \ref{ex:PES-isom}, we have that the possible transitions
for $\E$ are:
$$
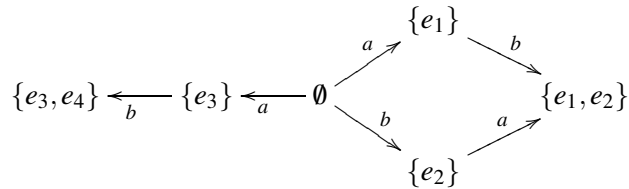
\xymatrix@R=12pt{
&&& \{e_1\} \ar[dr]^b
\\
\{e_3,e_4\} & \ar[l]^b \{e_3\} & \ar[l]^a \emptyset \ar[ur]^a\ar[dr]^b && \{e_1,e_2\}
\\
&&& \{e_2\} \ar[ur]^a
}
$$ 

The two most basic equivalences we shall consider are derived from process algebras and are
{\em bisimulation} and {\em trace equivalence}. To define the latter, we use the notion of
{\em (sequential) trace} of a PES $\E$, that is a sequence $a_1\ldots a_k \in \A^*$ such that
there exist $X_0,\ldots,X_k \in \Conf\E$ such that $X_0 = \emptyset$ and
$X_i \red {a_{i+1}} X_{i+1}$, for every $i = 0,\ldots,k-1$. We denote with $\STr \E$ the set
of the sequential traces of $\E$.

\begin{definition}[Interleaving Trace Equivalence \cite{CSP78}]
$\E \iteq \F$ if $\STr\E = \STr\F$.
\end{definition}

\begin{definition}[Interleaving Bisimulation \cite{Mil89}]
A relation $R \subseteq \Conf\E \times \Conf\F$ is an {\em interleaving bisimulation}
beween $\E$ and $\F$ if
\begin{itemize}
\item $(\emptyset;\emptyset) \in R$;
\item if $(X,Y) \in R$ and $X \red a X'$, then $Y \red a Y'$, for some $Y'$ such that $(X',Y') \in R$; and
\item if $(X,Y) \in R$ and $Y \red a Y'$, then $X \red a X'$, for some $X'$ such that $(X',Y') \in R$.
\end{itemize}
$\E \ibeq \F$ if there is an interleaving bisimulation between $\E$ and $\F$.
\end{definition}

Transitions involving a single action can be generalized to {\em steps}, i.e. sets of events that can be
executed simultaneously. 
Again, for the sake of abstraction, a step transition will be labeled with the multiset of labels associated
to the chosen concurrent events. Formally, we write $X \red A X'$ if $X \subset X'$, $X' \setminus X = G$,
$\forall e, e' \in G. e\ co\ e'$,
and $A$ is the multiset over $\A$ formed by the labels
of the events in $G$. 
For example, for $\E$ in Example  \ref{ex:PES-isom},
we now also have that $\emptyset \red{\{a,b\}} \{e_1,e_2\}$.
This yields the obvious generalization of interleaving bisimulation and trace 
equivalence, where step traces of $\E$, written $\StTr\E$, are defined as expected (i.e., like sequential traces, 
but with steps in place of single events).

\begin{definition}[Step Trace Equivalence \cite{Pomello85}]
$\E \steq \F$ if $\StTr\E = \StTr\F$.
\end{definition}

\begin{definition}[Step Bisimulation \cite{Pomello85}]
A relation $R \subseteq \Conf\E \times \Conf\F$ is a {\em step bisimulation}
beween $\E$ and $\F$ if
\begin{itemize}
\item $(\emptyset;\emptyset) \in R$;
\item if $(X,Y) \in R$ and $X \red A X'$, then $Y \red A Y'$, for some $Y'$ such that $(X',Y') \in R$; and
\item if $(X,Y) \in R$ and $Y \red A Y'$, then $X \red A X'$, for some $X'$ such that $(X',Y') \in R$.
\end{itemize}
$\E \sbeq \F$ if there exists a step bisimulation between $\E$ and $\F$.
\end{definition}

Because of their definition, configurations are actually partially ordered sets (posets, for short), 
where the ordering is given by $\leq$. Indeed, we write $\poset X$ to denote the labeled poset $(X, \leq|_X, l|_X)$, 
where $\leq|_X$ and $l|_X$ are the restrictions of $\leq$ and $l$ to $X$.
A more abstract view of a run is obtained by replacing events with their labels. This turns a poset 
into a partially ordered multiset ({\em pomset}, for short). Formally, the pomset associated to a 
configuration $X$, written $\pom X$, is the isomorphism class of $\poset X$.
We can then observe not just multisets, but multisets together with
their ordering, i.e. pomsets; this generalizes the step semantics because, by observing pomsets,
we are allowed to observe in one single transition also events that are not concurrent.
To this aim, we denote with $\POM \E$ the set of all pomsets of $\E$
and we label a transition with a pomset $p$, where $X \red p X'$ if $X \subset X'$, $X' \setminus X = H$
and $p = \pom H$.
Always referring to $\E$ in Example  \ref{ex:PES-isom},
we also have that $\emptyset \red{a;b} \{e_3,e_4\}$.

\begin{definition}[Pomset Trace Equivalence \cite{BoudolC87}]
$\E \pteq \F$ if $\POM\E = \POM\F$.
\end{definition}

\begin{definition}[Pomset Bisimulation \cite{BoudolC87}]
A relation $R \subseteq \Conf\E \times \Conf\F$ is a {\em pomset bisimulation}
beween $\E$ and $\F$ if
\begin{itemize}
\item $(\emptyset;\emptyset) \in R$;
\item if $(X,Y) \in R$ and $X \red p X'$, then $Y \red p Y'$, for some $Y'$ such that $(X',Y') \in R$; and
\item if $(X,Y) \in R$ and $Y \red p Y'$, then $X \red p X'$, for some $X'$ such that $(X',Y') \in R$.
\end{itemize}
$\E \pbeq \F$ if there exists a pomset bisimulation between $\E$ and $\F$.
\end{definition}

An orthogonal way to generalise the interleaving bisimulation is to keep track of the
causal dependencies and only relate configurations with the same causal history.
This is done by requiring that the two configurations have isomorphic associated posets,
where we also denote poset isomorphism with $\cong$.

\begin{definition}[Weak History Preserving Bisimulation \cite{DDM87}]
A relation $R \subseteq \Conf\E \times \Conf\F$ is a {\em weak history preserving bisimulation}
beween $\E$ and $\F$ if
\begin{itemize}
\item $(\emptyset;\emptyset) \in R$, and
\item if $(X,Y) \in R$ then
\begin{itemize}
\item $\poset X \cong \poset Y$;
\item if $X \red a X'$, then $Y \red a Y'$, for some $Y'$ such that $(X',Y') \in R$;
\item if $Y \red a Y'$, then $X \red a X'$, for some $X'$ such that $(X',Y') \in R$.
\end{itemize}
\end{itemize}
$\E \whbeq \F$ if there exists a weak history preserving bisimulation between $\E$ and $\F$.
\end{definition}

A stronger requirement is that the isomorphism relating $\poset{X'}$ and $\poset{Y'}$ cannot
be arbitrary, but must extend the isomorphism relating $\poset{X}$ and $\poset Y$. This
leads to the following definition.

\begin{definition}[History Preserving Bisimulation \cite{DDM88,RT88}]
A relation $R \subseteq \Conf\E \times \Conf\F \times 2^{\Conf\E \times \Conf\F}$ is a 
{\em history preserving bisimulation} beween $\E$ and $\F$ if
\begin{itemize}
\item $(\emptyset;\emptyset;\emptyset) \in R$, and
\item if $(X,Y,f) \in R$ then
\begin{itemize}
\item $f$ is an isomorphism between $\poset X$ and $\poset Y$;
\item if $X \red a X'$, then $Y \red a Y'$, for some $Y'$ such that $(X',Y',f') \in R$, where
$f'|_X = f$; and
\item if $Y \red a Y'$, then $X \red a X'$, for some $X'$ such that $(X',Y',f') \in R$, where
$f'|_X = f$.
\end{itemize}
\end{itemize}
$\E \hbeq \F$ if there exists a history preserving bisimulation between $\E$ and $\F$.
\end{definition}

The notion of history preserving bisimulation can be finally generalised by also asking for a `backwards'
bisimulation game, along the way of back-and-forth bisimulation \cite{DMV90}.

\begin{definition}[Hereditary History Preserving Bisimulation \cite{bedna01}]
A history preserving bisimulation $R$ beween $\E$ and $\F$ is {\em hereditary} if,
for every $(X,Y,f) \in R$, it holds that
$X' \red a X$ implies $(X',f(X'),f|_{X'}) \in R$
and 
$Y' \red a Y$ implies $(f^{-1}(Y'),Y',f|_{f^{-1}(Y')}) \in R$.

$\E \hhbeq \F$ if there exists a hereditary history preserving bisimulation between $\E$ and $\F$.
\end{definition}

All the equivalences presented so far form a well-known spectrum \cite{Fecher04,GG01}, 
depicted in Figure~\ref{eq-pes} 
(the only inclusions that are not present in \cite{GG01} are $\whbeq\ \subset\ \sbeq$
and $\whbeq\ \subset\ \pteq$, that are proved in \cite{Fecher04}).
There, the term {\em autoconcurrency} means existence of a configuration
containing two different concurrent events with the same label.

\section{Conflict without Causality: Coherence Spaces}

We now consider the first restriction of PESs, obtained by considering an empty causality relation.
This leads to {\em Coherence Spaces} \cite{Girard87}, a model largely studied, e.g., in the field of linear logic
and in the semantics of typed lambda-calculus \cite{BE91,BE91b,Girard87}.

\begin{definition}
A {\em coherence space} (written \CS) over an alphabet $\A$ is a PES $\E$ where the causality relation is empty.
\end{definition}

Thus, we shall usually omit $\leq$ from the definition of a \CS.
In the setting of \CS s, several definitions are radically simplified.
For example, a configuration is simply a finite and conflict-free subset of $E$;
similarly, two events 
are concurrent if they are not in conflict.
Moreover, a step and a pomset are simply multisets and, hence, the two notions do coincide.

Consequently, the spectrum of Figure~\ref{eq-pes} can be simplified, but it is still not trivial.
Indeed, removing the causality relation reduces a complex lattice to a simple chain:
trace equivalences all coincide and represent the coarsest notion; 
they properly include bisimulations (that all coincide, except for $\hhbeq$)
that in turn properly include the back-and-forth variant \cite{DMV90} of $\hbeq$
and the latter is still strictly coarser than isomorphism.
The spectrum is depicted in Figure~\ref{eq-sharp} and it is the first main result of this paper;
the following propositions are needed to establish it.

\begin{proposition}
\label{prop:wh}
For \CS s, if $\E \ibeq \F$ then $\E \hbeq \F$.
\end{proposition}
\begin{proof}
Let $R$ be an interleaving bisimulation between $\E$ and $\F$, and consider the
following relation:
$$
R' \define \bigcup_{(X,Y) \in R} \{(X,Y,f)\ :\ f \mbox{ is an isomorphism between }
X \mbox{ and } Y\}
$$
Trivially, $(\emptyset,\emptyset, \emptyset) \in R'$, because $(\emptyset,\emptyset) \in R$ 
and every set is isomorphic to itself. Let $(X,Y,f) \in R'$; by construction, $f:X \rightarrow Y$
is a bijection such that $l(x) = l(f(x))$, for every $x \in X$.\footnote{
	Indeed, notice that, for \CS s, the poset associated to a configuration is just a collection
	of (labeled) events (i.e., the ordering relation is empty) and, hence, poset isomorphism has
	only to respect the labeling.	
} Now, let $X \red a X'$; this means that $X' = X \uplus \{e\}$ and $l(e) = a$. Since
$(X,Y) \in R$, there exists $Y'$ such that $Y \red a Y' = Y \uplus \{e'\}$, where $l(e') = a$, and $(X',Y') \in R$.
It is easy to see that $f' = f \cup \{(e,e')\}$ is an isomorphism between $X'$ and $Y'$
and so $(X',Y',f') \in R$.
\end{proof}

%

\begin{proposition}
\label{prop:ipt}
For \CS s, if $\E \iteq \F$ then $\E \pteq \F$.
\end{proposition}
\begin{proof}
Since a pomset is just a multiset (i.e., a step), it suffices to prove that $\E \iteq \F$ implies $\E \steq \F$.
Let $A_1\ldots A_k \in \StTr \E$; we have to show that $A_1\ldots A_k \in \StTr \F$.

By definition, there exist $X_0,\ldots,X_k \in \Conf\E$ such that $X_0 = \emptyset$ and
$X_{i-1} \red {A_i} X_i$, for every $i = 1,\ldots,k$. This means that 
$X_{i-1} \subset X_i$, $X_i \setminus X_{i-1} = \{e^i_1,\ldots,e^i_{j_i}\}$,
$\forall h \neq q.\ e^i_h\ co\ e^i_q$ and $A_i$ is the multiset formed by $l(e^i_1),\ldots,l(e^i_{j_i})$.
Thus, $X_{i-1} \red {l(e^i_1)} \ldots \red {l(e^i_{j_i})} X_i$ and so
$l(e^1_1) \ldots l(e^1_{j_1}) \ldots l(e^k_1) \ldots l(e^k_{j_k}) \in \STr\E$.
By hypothesis, $l(e^1_1) \ldots l(e^1_{j_1}) \ldots l(e^k_1) \ldots l(e^k_{j_k}) \in \STr\F$;
i.e., there exist $Y_0,\ldots,Y_{j_1+\ldots+j_k} \in \Conf\F$ such that
$Y_0 = \emptyset$ and $Y_0 \red{l(e^1_1)} Y_1 \ldots \red{l(e^k_{j_k})} Y_{j_1+\ldots+j_k}$.
Since $Y_{j_1+\ldots+j_k}$ is a configuration and configurations in \CS s are
conflict-free sets, we have that all the events occurring in it are concurrent.
Thus, we can group single transitions into steps and obtain $A_1\ldots A_k \in \StTr \F$.
\end{proof}

%
%
%
%
%
%

\begin{proposition}
\label{prop:u-t-b}
There exist \CS s $\E$ and $\F$ such that $\E \iteq \F$ but $\E \not\ibeq \F$.
\end{proposition}
\begin{proof}
Consider
$\E = a + (a\parallel a)$ and $\F = a\parallel a$: they have the same traces (viz., $\{\epsilon,a,aa\}$)
but $\E$, after the leftmost $a$, is stuck, whereas $\F$, after every $a$, is not.
\end{proof}

\begin{proposition}
\label{prop:u-h-hh}
There exist \CS s $\E$ and $\F$ such that $\E \hbeq \F$ but $\E \not\hhbeq \F$.
\end{proposition}
\begin{proof}
Consider
$$
\E = a\parallel (a + (a\parallel a)) 
\qquad\text{ and } \qquad
\F = (a\parallel (a + (a\parallel a))) + (a\parallel a)
$$ 
and their LTSs (the events have been numbered in increasing order, from left to right, both in $\E$ and in $\F$):
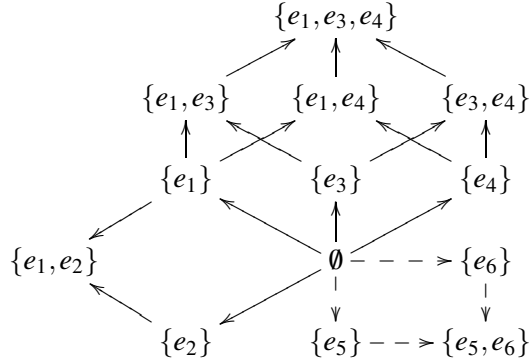
$$
\xymatrix@R=14pt@C=10pt{
&& \{e_1, e_3, e_4\} &
\\
& \{e_1, e_3\} \ar[ur] & \{e_1, e_4\} \ar[u] & \{e_3, e_4\} \ar[ul]
\\
& \{e_1\} \ar[u] \ar[ur]  \ar[dl] & \{e_3\} \ar[ur] \ar[ul] & \{e_4\}  \ar[ul] \ar[u]
\\
\{e_1, e_2\} && \emptyset \ar[ur] \ar[ul] \ar[u] \ar[dl] \ar@{-->}[d] \ar@{-->}[r] & \{e_6\}  \ar@{-->}[d]
\\
& \{e_2\}  \ar[ul] & \{e_5\}  \ar@{-->}[r] & \{e_5, e_6\}
}
$$
Here, states are configurations, arrows are $a$-labeled transitions and the LTS for $\E$
is the solid part, whereas the LTS for $\F$ also includes the dashed part.

The only possible history preserving bisimulation between $\E$ and $\F$ is the one that acts as the identity on the
common configurations and that associates $\{e_5, e_6\}$ with $\{e_1, e_2\}$ and both $\{e_5\}$ 
and $\{e_6\}$ with $\{e_2\}$. 
However, it is not hereditary because from $\{e_1, e_2\}$ we can backtrack to $\{e_1\}$
and from here we can perform two $a$'s in sequence; by contrast, every backtrack from $\{e_5, e_6\}$
leads to a configuration that can only perform one single $a$.
\end{proof}

\begin{proposition}
\label{prop:hhp-isom}
There exist \CS s $\E$ and $\F$ such that $\E \hhbeq \F$ but $\E \not\cong \F$.
\end{proposition}
\begin{proof}
The example given in \cite{bedna01} for proving a similar claim 
(viz., $\E = a$ and $\F = a+a$) is in fact made up from two CSs.
\end{proof}

Quite surprisingly, the proofs of Propositions~\ref{prop:wh} and \ref{prop:ipt} do not rely on the fact that
labels are different or not, and the examples provided in Propositions~\ref{prop:u-t-b}, \ref{prop:u-h-hh}
and~\ref{prop:hhp-isom} are built on \CS s where all events have the same label.
Hence, in the setting of \CS s, the labeling function has no impact on
the spectrum of Figure~\ref{eq-sharp}.

%

\section{Causality without Conflict: Elementary ESs}

A second restriction of PESs is obtained by considering an empty conflict relation;
this yields {\em Elementary Event Structures} \cite{NielsenPW81}.

\begin{definition}
An {\em elementary event structure} (written EES) over an alphabet $\A$ is a 
PES $\E$ where the conflict relation is empty.
\end{definition}

Consequently, we shall omit $\sharp$ from the definition of an EES.
EESs are a particular kind of directed acyclic graphs, where
every path from $u$ to $v$ entails the existence of a directed edge $(u,v)$;
this comes from the fact that causality is transitive. For the sake of simplicity, we shall sometimes
represent EESs with the transitive reduction\footnote{
	The \emph{transitive reduction} 
	of a DAG $D$ is the (unique) smallest DAG $D'$ which preserves the reachability relation of $D$. 
	Note that two transitively reduced DAGs are isomorphic if and only if
	their transitive closures are isomorphic.} 
of their causality relation.
For example,
$$
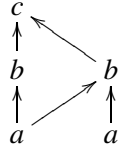
\xymatrix@R=12pt{
c
\\
b \ar[u]& b\ar[ul]
\\
a \ar[u]\ar[ur] & a \ar[u]
}
$$
represents the (isomorphism class of the) EES $\E=(E,\leq,l)$, where $E=\{e_1,e_2,e_3,e_4,e_5\}$,
$l(e_1)=l(e_2)=a$, $l(e_3)=l(e_4)=b$, $l(e_5)=c$, $e_i \leq e_i$, $e_1 \leq e_3$, $e_1 \leq e_4$, $e_1 \leq e_5$, 
$e_2 \leq e_4$, $e_2 \leq e_5$, $e_3 \leq e_5$ and $e_4 \leq e_5$.

We now present the results needed to adapt the spectrum of Figure~\ref{eq-pes}
to EESs; this is the second main contribution of our work.
Surprisingly, the spectrum changes according to whether the set of events is finite or not.
However, there are a few common results, that we now present.

For interleaving and step equivalences, the spectrum for EESs is the same as
that for PESs: the inclusions depicted in the upper part of Figure~\ref{eq-pes} also hold for EESs; 
what changes are the counterexamples needed to distinguish them. We now provide the
distinguishing examples in the framework of EESs.

\begin{proposition}
\label{prop:i-s}
For EESs, there exist $\E$ and $\F$ such that
$\E \ibeq \F$ and $\E \iteq \F$, whereas $\E \not\sbeq \F$ and $\E \not\steq \F$.
\end{proposition}
\begin{proof}
Consider the EESs $\E = a;a$ and $\F = a \parallel a$.
\end{proof}

\begin{proposition}
\label{prop:it-ib}
For EESs, there exist $\E$ and $\F$ such that
$\E \iteq \F$, whereas $\E \not\ibeq \F$. 
\end{proposition}
\begin{proof}
Consider the EESs $\E = (a \parallel b);(a \parallel b)$ and $\F = (a;b) \parallel (b;a)$.
Trivially, $\E \iteq \F$, since $\STr\E = \STr\F = \{\epsilon, a, b, ab, ba, aba, abb,$ $ baa, bab, abab, abba, baab, baba\}$.
By contrast $\E \not\ibeq \F$, since in $\F$ we can reach, after executing the leftmost $a$ and $b$, a state 
where only $b$ is possible, whereas in $\E$, after every $a$ and $b$, 
both $a$ and $b$ are always enabled.
\end{proof}

\begin{proposition}
\label{prop:st-i}
For EESs, there exist $\E$ and $\F$ such that
$\E \steq \F$ whereas $\E \not\ibeq \F$.
\end{proposition}
\begin{proof}
Consider the EESs
$$
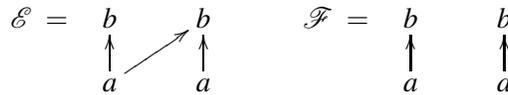
\begin{array}{cccc}
\E \ = &
\xymatrix@R=12pt{
b & b
\\
a \ar[u]\ar[ur] & a \ar[u]
}
&
\qquad \F \ = &
\xymatrix@R=12pt{
b & b
\\
a \ar[u] & a \ar[u]
}
\end{array}
$$
The step LTSs resulting from these EESs (where states are configurations and arrows represent
transitions) are:
$$
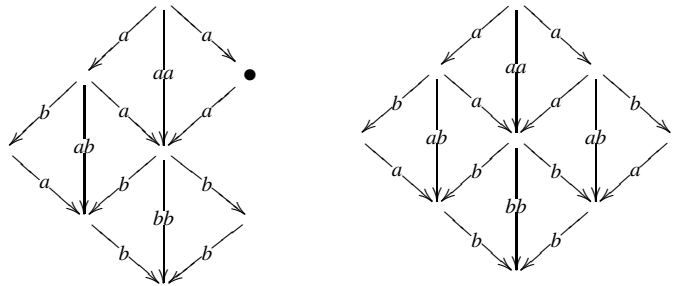
\begin{array}{cc}
\xymatrix@R=20pt{
&& \ar[dl]|{a} \ar[dr]|a \ar[dd]|{aa}
\\
& \ar[dl]|b \ar[dr]|a \ar[dd]|{ab}
&&
\bullet \ar[dl]|a
\\
\ar[dr]|a
&&
\ar[dl]|b \ar[dr]|b \ar[dd]|{bb}
\\
& \ar[dr]|b
&&
 \ar[dl]|b
\\
&&
}
&
\qquad
\xymatrix@R=20pt{
&& \ar[dl]|a \ar[dr]|a \ar[dd]|{aa}
\\
& \ar[dl]|b \ar[dr]|a \ar[dd]|{ab}
&&
\ar[dl]|a \ar[dr]|b \ar[dd]|{ab}
\\
\ar[dr]|a
&&
\ar[dl]|b \ar[dr]|b \ar[dd]|{bb}
&&
\ar[dl]|a
\\
& \ar[dr]|b
&&
 \ar[dl]|b
\\
&&
}
\end{array}
$$
From them, checking $\steq$ is immediate.
On the other hand, $\ibeq$ does not hold because
there exists a configuration in the left-hand side LTS (marked with `$\bullet$')
reachable after an $a$ that cannot perform a $b$; by contrast, every configuration reachable
after an $a$ in the right-hand side LTS can always perform a $b$.
\end{proof}

An easy corollary of the previous result is that, for EESs, $\steq$ is not contained
in $\pteq$ and $\sbeq$.

\begin{proposition}
\label{prop:s-hh}
For EESs, there exist $\E$ and $\F$ such that
$\E \sbeq \F$ whereas $\E \not\whbeq \F$ and $\E \not\pteq \F$ .
\end{proposition}
\begin{proof}
Consider the EESs
$$
\begin{array}{cccc}
\E \ = &
\xymatrix@R=12pt@C=5pt{
&&&e_7
\\
&&&e_6
\\
e_2 \ar[uurrr] && e_3 \ar[ur] && e_4 \ar[ul] && e_5 \ar[uulll]
\\
&e_0 \ar[ul]\ar[ur]\ar[urrr] &&&& e_1 \ar[ulll]\ar[ul]\ar[ur]
}
&
\qquad \F \ = &
\xymatrix@R=24pt@C=5pt{
&e_6'&&&&e_7'
\\
e_2' \ar[ur] && e_3' \ar[ul] && e_4' \ar[ur] && e_5' \ar[ul]
\\
&e_0' \ar[ul]\ar[ur]\ar[urrr] &&&& e_1' \ar[ulll]\ar[ul]\ar[ur]
}
\end{array}
$$
where all events $e_0,\ldots,e_7,e_0',\ldots,e_7'$ have the same label.
It can be readily checked that $\E \not\cong \F$ because events $e_6$ and $e_7$ in $\E$ have no
isomorphic correspondence in $\F$. Thus, trivially, $\E$ and $\F$ cannot be in $\pteq$; moreover, they cannot either be in $\whbeq$ because
every weak history preserving bisimulation must contain the pair $(E,F)$, but this is not possible
since $\E \not\cong \F$.

By contrast, we shall now prove that $\E \sbeq \F$.
To this aim, for a generic EES $\E=(E,\leq,l)$ and for every $X \in \Conf\E$, we denote with 
$\E_X$ the EES
$(E \setminus X, \leq|_{E \setminus X} , l |_{E \setminus X})$.
It is now easy to check that:
\begin{itemize}
\item $\E_{\{e_0\}} \cong \F_{\{e_0'\}}$, via some isomorphism $f_{e_0,e_0'}$ (for example: $(e_1,e_1'),(e_2,e_2'),(e_3,e_4'),(e_4,e_5'),(e_5,e_3'),$ $(e_6,e_7'),(e_7,e_6')$);
\item $\E_{\{e_1\}} \cong \F_{\{e_1'\}}$, via some isomorphism $f_{e_1,e_1'}$ (for example: $(e_0,e_0'),(e_2,e_4'),(e_3,e_2'),(e_4,e_3'),(e_5,e_5'),$ $(e_6,e_6'),(e_7,e_7')$);
\item $\E_{\{e_0,e_1\}} \cong \F_{\{e_0',e_1'\}}$, via some isomorphism $f_{e_0e_1,e_0'e_1'}$ (for example: $(e_2,e_2'),(e_3,e_4'),(e_4,e_5'),(e_5,e_3'),$ $(e_6,e_7'),(e_7,e_6')$).
\end{itemize}
Notationally, let $f(X)$ denote $\{f(x)\ :\ x \in X\}$;
it can be now verified that the relation
$$
\begin{array}{ll}
R \ = & \{(\emptyset,\emptyset)\}\ \cup
\vspace*{.1cm}
\\
& \cup\ \bigcup_{X \in \Conf{\E_{\{e_0\}}}}  \{ (  \{e_0\}\cup X , \{e_0'\}\cup f_{e_0,e_0'}(X) ) \}
\vspace*{.1cm}
\\
& \cup\ \bigcup_{X \in \Conf{\E_{\{e_1\}}}}  \{ (  \{e_1\}\cup X , \{e_1'\}\cup f_{e_1,e_1'}(X) ) \}
\vspace*{.1cm}
\\
& \cup\ \bigcup_{X \in \Conf{\E_{\{e_0,e_1\}}}}  \{ (  \{e_0,e_1\}\cup X , \{e_0',e_1'\}\cup f_{e_0e_1,e_0'e_1'}(X) ) \}
\end{array}
$$
is a step bisimulation between $\E$ and $\F$.
\end{proof}

An easy corollary of the previous result is that, for EESs, $\sbeq$ is not contained
in $\pbeq$, $\hbeq$, $\hhbeq$, and $\cong$.
Furthermore, notice that the examples provided in Propositions~\ref{prop:i-s} and~\ref{prop:s-hh}
use EESs with a ``flattening'' labeling function (mapping all events to the same label);
by contrast, this is not the case in Propositions  \ref{prop:it-ib} and \ref{prop:st-i}. This is 
not incidental, since,
for EESs with all events labeled the same, $\ibeq$ and $\iteq$ coincide; to prove this,
we first need a lemma.

\begin{lemma}
\label{lem:nored}
Let $\E = (E,\leq,l)$ be an EES 
and let $X \in \Conf\E$; then either $X\ \red{l(e)} X\cup\{e\}$, for some $e \in E \setminus X$, or $X = E$.
\end{lemma}
\begin{proof}
Let $e \in E$; by induction on $|\{e'\ :\ e' < e\}|$, we prove that either $e \in X$ 
or there exists an $e' \leq e$ such that
$X \red{l(e')} X \cup \{e'\}$. The base case is trivial. For the inductive case, let us assume $e$ with at least one
predecessor. If $e \in X$, we are done. If $X$ contains all the predecessors of $e$ (but not $e$), then
$X\ \red{l(e)} X\cup\{e\}$. Otherwise, consider any $e' < e$ not contained in $X$; the claim follows by induction,
since $e'$ has less predecessors than $e$ (indeed, every predecessor of $e'$ is also a predecessor of $e$).
\end{proof}

\begin{proposition}
For EESs with labeling set ${\cal A}=\{a\}$, $\ibeq=\,\iteq$.
\end{proposition}
\begin{proof}
Lemma \ref{lem:nored} entails that $\STr\E$ is $\{a^n\,:\, 0 \leq n \leq |E|\}$, if $E$ is finite, or
$\{a^n\,:\, n \geq 0\}$, otherwise. The same holds for $\F$; hence, if $\E \iteq \F$, then $|E|=|F|$.

Now, let $\E \iteq \F$ and $R = \{(X,Y)\, :\, X \in \Conf\E, Y \in \Conf\F, |X|=|Y|\}$.
Trivially, $(\emptyset,\emptyset) \in R$. Furthermore, if $(X,Y) \in R$ and $X \red a X'$, then
$|Y| = |X| < |E| = |F|$; again by Lemma \ref{lem:nored}, 
there exists $Y \red a Y'$ and, by construction, $(X',Y') \in R$.
\end{proof}

Hence, differently from \CS s, in the framework of EESs 
the labeling function has an impact on the distinguishing
power of the equivalences studied.

To complete the hierarchy of equivalences for EESs, we surprisingly discovered that there is a deep
difference if we consider finite or infinite event structures. For the former ones, we have been able to
completely define the spectrum; for the latter ones, we still have some open questions, mostly on the
history preserving bisimulations.

\subsection{Finite EESs}


For finite EESs, we have the following results
that lead to the spectrum in Figure~\ref{fig:finiteEESs}.

\begin{proposition}
\label{prop:pteq-isom}
Let $\E$ and $\F$ be finite EESs such that $\E \pteq \F$; then $\E \cong \F$.
\end{proposition}
\begin{proof}
The key observation is that, in every finite EES $\E$, the set of all the events $E$ is
a configuration (it is finite, conflict-free and closed by causal predecessors). Hence,
$Pomset(E) = \E$. So, if $\E \pteq \F$, it holds that $\E$ and $\F$ have the same pomsets;
in particular, the pomsets corresponding to $E$ and $F$ must be the same. Hence, the two EESs
are isomorphic.
\end{proof}

\begin{proposition}
For finite EESs, $\cong\ =\ \hhbeq\ =\ \hbeq\ =\ \whbeq\ =\ \pbeq\ =\ \pteq$.
\end{proposition}
\begin{proof}
For all equivalences but $\whbeq$ the claim is an
easy corollary of the previous proposition, by the fact that $\cong\ \subseteq\ \pteq$ for PESs
(and, hence, also for EESs). For $\whbeq$, take $\E \whbeq \F$ and consider a sequence of transitions
$\emptyset \red{a_1} \ldots \red{a_n} E$ (whenever $|E| = n$). The only possible reply to this sequence
is some $\emptyset \red{a_1} \ldots \red{a_n} F'$ such that $F' = F$, otherwise $F' \red{}$ whereas
$E \ \not\!\!\red{}$. Thus, $\E \cong\F$, since $\E = \poset E$ and $\F = \poset F$.
\end{proof}


\subsection{Infinite EESs}

For infinite EESs, we first notice that, if we consider EESs of different cardinality, 
Proposition~\ref{prop:pteq-isom} does not hold.
To see this, consider $\E$ and $\F$ made up, respectively, by a numerable and by a non-numerable
set of concurrent copies of the same pomset; clearly, the two structures have the same (finite) pomsets
and, hence, are pomset trace equivalent, but of course they are not isomorphic.

Moreover, Proposition~\ref{prop:pteq-isom} does not hold either for EESs of the same cardinality,
as the following Propositions entail.

\begin{proposition}
There exist $\E$ and $\F$ infinite EESs such that $\E \pbeq \F$, $\E \not\hbeq \F$ and $\E \not\whbeq \F$.
\end{proposition}
\begin{proof}
Assume two numerable sets of events, $E = \{e_{ij}\}_{0\leq j\leq i}$ and $F =  \{e'_{ij}\}_{i,j \geq 0}$.
Let $\leq_E$ (resp., $\leq_F$) be such that $e_{ij} \leq_E e_{ik}$ (resp., $e'_{ij} \leq_F e'_{ik}$) if and only if $j \leq k$.
Finally, let every event be labeled with the same label $a$, both in $\E$ and in $\F$.
Pictorially:
$$
\begin{array}{cc}
\xymatrix@R=12pt@C=10pt{
&
\\
&&& e_{22}
\\ 
\E \ = && e_{11} & e_{21}\ar[u]
& \ldots
\\
& e_{00} & e_{10} \ar[u] & e_{20} \ar[u]
}
&
\qquad
\xymatrix@R=12pt@C=10pt{
& \cdots & \cdots & \cdots 
\\
& e'_{02}\ar[u] & e'_{12}\ar[u] & e'_{22}\ar[u]
\\
\F \ = & e'_{01} \ar[u] & e'_{11} \ar[u] & e'_{21} \ar[u] & \ldots
\\
& e'_{00}\ar[u] & e'_{10}\ar[u] & e'_{20}\ar[u]
}
\end{array}
$$

To show that $\E \not\hbeq \F$ and $\E \not\whbeq \F$, consider $\emptyset \red a \{e_{00}\}$:
the only possible reply in $\F$ is $\emptyset \red a \{e'_{i0}\}$, for some $i$. However, 
$\{e_{00}\}$ and $\{e'_{i0}\}$ cannot be related by any history or weak history preserving bisimulation:
indeed, the challenge $\{e'_{i0}\} \red a \{e'_{i0},e'_{i1}\}$ has no possible reply, since there
is no event in $E$ causally dependent on $e_{00}$ (whereas $e'_{i0} <_F e'_{i1}$).

To prove that $\E \pbeq \F$, consider
\[
\begin{array}{ccl}
R_0  & = &  \{ (\emptyset, \emptyset) \} \\
\\
R_{n+1}  & = & \{ (X',Y')\ :\ \exists (X,Y) \in R_n\exists p\,.\ X \red p X' \wedge Y \red p Y' \} \\
\\
R  & = &   \bigcup_{n \geq 0} R_n
\end{array}
\]
We now prove that $R$ is a pomset bisimulation. By construction, $(\emptyset, \emptyset) \in R$.
Let $(X,Y) \in R$. 
If $X \red p X'$, then $p$ is a finite collection of finite chains (w.r.t. $\leq_E$) and, hence, can
be embedded into $\{e'_{ij}\}_{i>n,j\geq 0}$, where $n$ is the largest integer such that $e'_{n0} \in Y$.
Let $\hat Y \subset \{e'_{ij}\}_{i>n,j\geq 0}$ be such that $Pomset(\hat Y) = p$; then, $Y \red p Y \uplus \hat Y = Y'$
and $(X',Y') \in R$ by construction.
If $Y \red p Y'$, then $p$ is a finite collection of finite chains (w.r.t. $\leq_F$); let $h$ be the shortest of such chains. Now, $p$ can be embedded into 
$\{e_{ij}\}_{i>m,0\leq j\leq i}$, where $m = \max\{h,n\}$ and $n$ is the largest integer such that $e_{n0} \in X$.
Let $\hat X \subset \{e_{ij}\}_{i>m,0\leq j\leq i}$ be such that $Pomset(\hat X) = p$; 
then, $X \red p X \uplus \hat X = X'$ and $(X',Y') \in R$ by construction.
\end{proof}

\begin{proposition}
There exist $\E$ and $\F$ infinite EESs such that $\E \pteq \F$ but $\E \not\ibeq \F$.
\end{proposition}
\begin{proof}
Let us consider 
$$
\begin{array}{cc}
\xymatrix@R=12pt@C=10pt{
&& b & b & b
\\
\E \ = & a & a\ar[u] & a\ar[u] & a\ar[u]
& \ldots
}
&
\qquad
\xymatrix@R=12pt@C=10pt{
& b & b & b
\\
\F \ = & a \ar[u] & a \ar[u] & a \ar[u] & \ldots
}
\end{array}
$$
We have that  $\E \pteq \F$ because
$$
pomsets(\E) = pomsets(\F) = 
\left\{
\begin{array}{c}
\underbrace{
\xymatrix@R=10pt@C=1pt{
&&&
\\
a & a & \ldots & a
}
}_m\ 
\underbrace{
\xymatrix@R=10pt@C=1pt{
b & b && b
\\
a \ar[u] & a \ar[u] & \ldots & a \ar[u]
}
}_n
\end{array}
\right\}_{m,n \geq 0}
$$
By contrast, the singleton $a$ in $\E$ cannot be replied to by any $a$ in $\F$ because the
latter enables a $b$, whereas the former does not.
\end{proof}

An easy corollary of the last Proposition is that $\pteq$ does not imply $\whbeq$, $\pbeq$ and $\sbeq$.
We conclude this section by a list of questions that remain to be answered.

\bigskip
\noindent{\bf Open questions:}
Are there $\E$ and $\F$ (infinite EESs) such that
\begin{enumerate}
\item  $\E \whbeq \F$ but $\E \not\pbeq \F$?
\item  $\E \whbeq \F$ but $\E \not\hbeq \F$? 
\item  $\E \hbeq \F$ but $\E \not\hhbeq \F$?
\item  $\E \hhbeq \F$ but $\E \not\cong \F$?
\end{enumerate}
If all these open questions have a positive answer, the spectrum for infinite EESs,
depicted in Figure~\ref{infEESs},
is the same as the one for general PESs, depicted in Figure~\ref{eq-pes}.
Notice that, if open question 1 has a positive answer, the same holds also for open question 2.
However, we conjecture that also in the setting of infinite EESs all history-preserving 
bisimulation equivalences coincide and coincide with isomorphism; however, we still do not have
enough evidences for formally proving this claim.

\section{Conclusion}

In this paper we studied how the spectrum of equivalences for PESs defined in
\cite{Fecher04,GG01} changes when alternatively removing causality and conflict.
In both cases, equivalences that
are properly included in one another for PESs turn out to coincide and this is more
evident in \CS s than in EESs.
Moreover, both the labeling function and the cardinality of the event set influence the spectrum
for EESs, whereas they have no impact on the spectrum for \CS s. 
For these reasons, we argue that causality is a more foundational building block than conflict in event structures,
since it has a deeper impact on the discriminating power of equivalences for such models and
because it is more sensitive than conflict to issues like the cardinality of the set of events and their labeling.

Surely, our results can be also related to the fact that the equivalences considered are causality-based (apart from the interleaving ones). Maybe, conflict could have a deeper impact than causality on other kinds of equivalences or on different models (for instance, variants of ESs with asymmetric choice, or with two different kinds of choices -- external and internal, or nondeterministic and probabilistic). This is a first interesting line for future research.

Another possible extension of our work is the investigation of other equivalences
(like, e.g., those presented in \cite[Sect.\,3]{GGS13}) and their impact on the spectra presented in this paper.
However, we do not believe that this would change the message conveyed by this paper. By contrast,
a challenging direction for future research would be the adaptation to \CS s and EESs of the logical characterizations
given by \cite{BC14} to the equivalences studied in this paper. In particular, it would be nice to see how the logical 
operators defined in \cite{BC14} can be simplified for capturing the equivalences in the simplified frameworks.

Finally, it is interesting to note that the pair of EESs in the proof of Proposition 
\ref{prop:s-hh} has been obtained through an exhaustive search on transitively 
reduced DAGs, using the tools in the \texttt{nauty/Traces} \cite{nautypage,McKay201494} 
distribution. More precisely, it is the smallest (with respect to number of vertices) 
pair of non-isomorphic transitively reduced DAGs having the same multiset of source-deleted subgraphs.

\bigskip
\noindent{\bf Acknowledgements} We are grateful to 
Silvia Crafa and Paolo Baldan for fruitful discussions and to 
Rob van Glabbeek for the counterexample of Proposition~\ref{prop:u-h-hh}.

\bibliographystyle{eptcs}
\bibliography{sharpDoi,ivanoDoi}

\end{document}